\documentclass{article}
\usepackage[margin=1.5in]{geometry}
\usepackage{IEEEtrantools}
\usepackage[english]{babel}
\usepackage{graphicx}
\usepackage{amssymb, latexsym, amsthm}
\usepackage[cmex10]{amsmath}
\usepackage{algpseudocode, algorithm}
\usepackage{cite}
\usepackage{url}
\usepackage{subfloat}
\usepackage{longtable}
\usepackage{xcolor,calc}
\usepackage[utf8]{inputenc}
\providecommand{\Ntruck}{N_{\textrm{truck}}}
\providecommand{\Ttruck}{T_{\textrm{truck}}}
\providecommand{\Tprice}{T_{\textrm{price}}}
\providecommand{\Tutil}{T_{\textrm{util}}}
\providecommand{\Thist}{T_{\textrm{hist}}}

\newtheorem{theorem}{Theorem}

\begin{document}

\title{Dynamic vehicle redistribution and online price incentives in shared mobility systems}

\author{Julius Pfrommer \and Joseph Warrington \and Georg Schildbach \and Manfred Morari}

\date{\today}

\maketitle

\begin{abstract}
  This paper considers a combination of intelligent repositioning decisions and
  dynamic pricing for the improved operation of shared mobility systems. The
  approach is applied to London's Barclays Cycle Hire scheme, which the authors
  have simulated based on historical data. Using model-based predictive control
  principles, dynamically varying rewards are computed and offered to customers
  carrying out journeys. The aim is to encourage them to park bicycles at nearby
  under-used stations, thereby reducing the expected cost of repositioning them
  using dedicated staff. In parallel, the routes that repositioning staff should
  take are periodically recomputed using a model-based heuristic. It is shown
  that it is possible to trade off reward payouts to customers against the cost
  of hiring staff to reposition bicycles, in order to minimize operating costs
  for a given desired service level.
\end{abstract}

\section{Introduction}
\label{sec:intro}

Public Bicycle Sharing (PBS) schemes offer the rental of bicycles as a means of
public transportation in urban areas. They allow registered users to pick up a
bicycle from one of many docking stations throughout the entire city, without
any prior notice. The bicycle is returned to another station, after which the
user's intended destination is usually reached on foot. Short journeys are
encouraged by charging users only a small fee for a short rental period
(typically less than one hour), but then ramping up the cost significantly for
longer use.

Such schemes have been introduced in major cities as an alternative to often
slow and crowded mass transportation. Many have grown considerably in size in
recent years \cite{shaheen_bikesharing_2010, demaio_bike-sharing:_2009}, and are
becoming a major component of their cities' public transportation systems. In
2008, for example, 120{,}000 daily journeys were being made using shared
bicycles in Paris \cite{erlanger2008}.

Most PBS schemes are still unable to recoup their full operational and
investment costs solely from customer fees. According to
\cite{midgley2011bicycle}, capital costs can be up to \$4,500 per bicycle, and
annual operational costs up to \$1,700 per bicycle. Sometimes additional
revenues from advertising can be used to mitigate this cost gap. However, in
almost all cases additional funding from public sources is required
\cite{shaheen_bikesharing_2010,wang2010bike}.

One of the major contributors to operational costs is the need to operate
staffed trucks for manual relocation of bicycles, in order to balance the
difference between supply and demand at various stations. If this effort were
not made, the arrival and departure of customers would cause many stations to
run full or empty, and the customer service rate would drop below acceptable
levels \cite{vogel_modeling_2010,obis_optimising_2011}. Since this
redistribution of bicycles entails costs, a trade-off for the desired service
level needs to be made.

The goal of this paper is to show how the system performance could be optimized
by trading off two complementary methods. Firstly, we devise an algorithm to
optimize the dynamic route-planning of multiple trucks for bicycle relocation.
Secondly, on top of this manual repositioning, we propose a scheme that offers
users price incentives based on the current and predicted state of the system, in
order to encourage them to change the endpoint of their journeys. These
incentives are set to shift bicycle drop-offs away from stations that are
overfilled, and towards nearby stations that may have empty spaces. The price
incentives are independent of the usual rental fees, which we assume to be a
sunk cost for the user.

This paper's main contributions can be summarized as follows:
\begin{enumerate}
\item A tailored routing algorithm that plans how trucks will be used to
  redistribute bikes amongst stations. Redistribution is performed in the
  dynamic setting, i.e. while the system is in operation. The heuristic
  chooses the actions of multiple trucks, with the aim of enabling as many extra
  journeys as possible to take place.
\item A dynamic incentives scheme where customers are encouraged to change their
  target station in exchange for a payment. Changes to journey length may be
  inconvenient, and we assume customers accept or reject such incentives based
  on the value of their time and the payment offered.
\end{enumerate}
Both the truck routes and the price incentives are recomputed online at periodic
time instances. For both components, a predictive model of the expected
near-future evolution of the system state is used to optimize their actions over
a finite, receding horizon. The optimization goal is to maximize the number of
additional journeys enabled via repositioning, taking into account available
resources and cost trade-offs. At each re-optimization step, up-to-date
information on the current state of the system is used to plan all future
operational decisions. This is shown schematically in Fig. \ref{fig:sysop}.

\begin{figure}
  \centering
  \includegraphics[width=.6\textwidth]{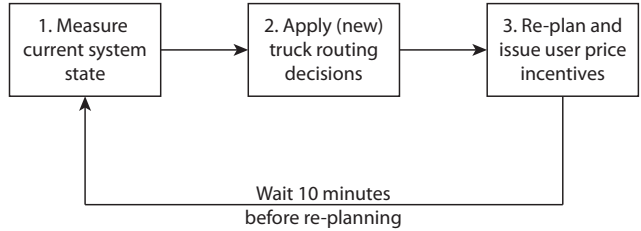}
  \caption{Schematic of the online optimization scheme presented in this paper.
    In step 2, truck routes are planned based on a model of how the bike
    movements will evolve from the current state, and if necessary new orders
    are issued. In step 3, the current system state and the future truck actions
    are taken into account, and new price incentives for users are computed,
    based on a trade-off between payouts and system performance.}
  \label{fig:sysop}
\end{figure}

We evaluate the trade-off between these two methods using a Monte Carlo model of
the London Cycle Hire PBS scheme, constructed from detailed historical usage
information. Our results suggest that service level improvements may be attained
using price incentives alone, and that increases in either the customer payouts
or the number of repositioning trucks deliver diminishing returns to service
levels. Unsurprisingly, higher service levels can be reached on weekends in
comparison to weekdays.

The paper is organized as follows. In Section~\ref{sec:model} we explain how a
model of the London PBS scheme was derived using historical data. In
Section~\ref{sec:utility} we develop a metric for the utility of repositioning
actions based on the expected ability to serve additional future customers. The
results are used in Section~\ref{sec:trucks} to develop a heuristic for
determining the routes of repositioning trucks. In Section~\ref{sec:pricing} we
develop a model-based controller for computing the price incentives offered to
customers. The two approaches for repositioning are compared using a Monte Carlo
simulation framework in Section~\ref{sec:simulation}.
Section~\ref{sec:conclusion} draws conclusions on the performance of the scheme.

\section{System model}
\label{sec:model}

\subsection{Historical data}
The PBS system model used in this paper is based on London's Barclays Cycle Hire
scheme. For modeling, we used three data sets made publicly available by the
Transport for London
authority:\footnote{\url{http://www.tfl.gov.uk/businessandpartners/syndication/default.aspx}}
\begin{enumerate}
\item 1{.}42 million rides spanning a period of 97 days (examples in
  Table~\ref{tab:ride-format}),
\item Size and location of 354 stations actively used during the recorded period
  (examples in Table~\ref{tab:stations-format}),
\item An initial station fill level recorded during nighttime when all bicycles
  were docked. In total, we estimate that the system contained 3708 bikes at
  the time for which historical data is available.
\end{enumerate}
\begin{table}
  \caption{Ride information samples}
  \label{tab:ride-format}
  \centering
  \begin{tabular}{c c c }
    \hline\noalign{\smallskip}
   bike-id & start (date, station-id) & end (date, station-id) \\
   \noalign{\smallskip}\hline\noalign{\smallskip}
    3340 & \{2010-07-30 06:00:00, 47\} & \{2010-07-30 06:22:00, 47\} \\
    3870 & \{2010-07-30 06:00:00, 234\} & \{2010-07-30 06:14:00, 203\} \\
    1627 & \{2010-07-30 06:01:00, 149\} & \{2010-07-30 06:29:00, 293\} \\
    1695 & \{2010-07-30 06:02:00, 152\} & \{2010-07-30 06:06:00, 324\} \\
    \noalign{\smallskip}\hline
  \end{tabular}
\end{table}
\begin{table}
  \caption{Station information samples}
  \label{tab:stations-format}
  \centering
  \begin{tabular}{c c c c c }
    \hline\noalign{\smallskip}
    id & name & position (lat, lon)& size\\
    \noalign{\smallskip}\hline\noalign{\smallskip}
    1 & River St, Clerkenwell & \{51.5291, -0.1099\} & 18\\
    2 & Phillimore Gardens, Kensington & \{51.4996, -0.1975\} & 34\\
    3 & Christopher St, Liverpool St & \{51.5212, -0.0846\} & 33\\
    4 & St. Chad's Street, King's Cross & \{51.5300, -0.1200\} & 22 \\
    \noalign{\smallskip}\hline
  \end{tabular}
\end{table}

Analysis of the historical journeys reveals regular daily flow patterns, with a
substantial difference between weekdays and weekends. Journeys are allowed
between 6am and midnight. As expected, many customers commute to the city center
in the morning and ride back to the outer districts in the late afternoon. This
pattern, absent on weekends, causes two spikes in daily rental activities that
are illustrated in Fig. \ref{fig:daily-rides}.
\begin{subfigures}
  \begin{figure}
    \begin{minipage}[h]{.475\linewidth}
      \centering
      \includegraphics[width=\textwidth]{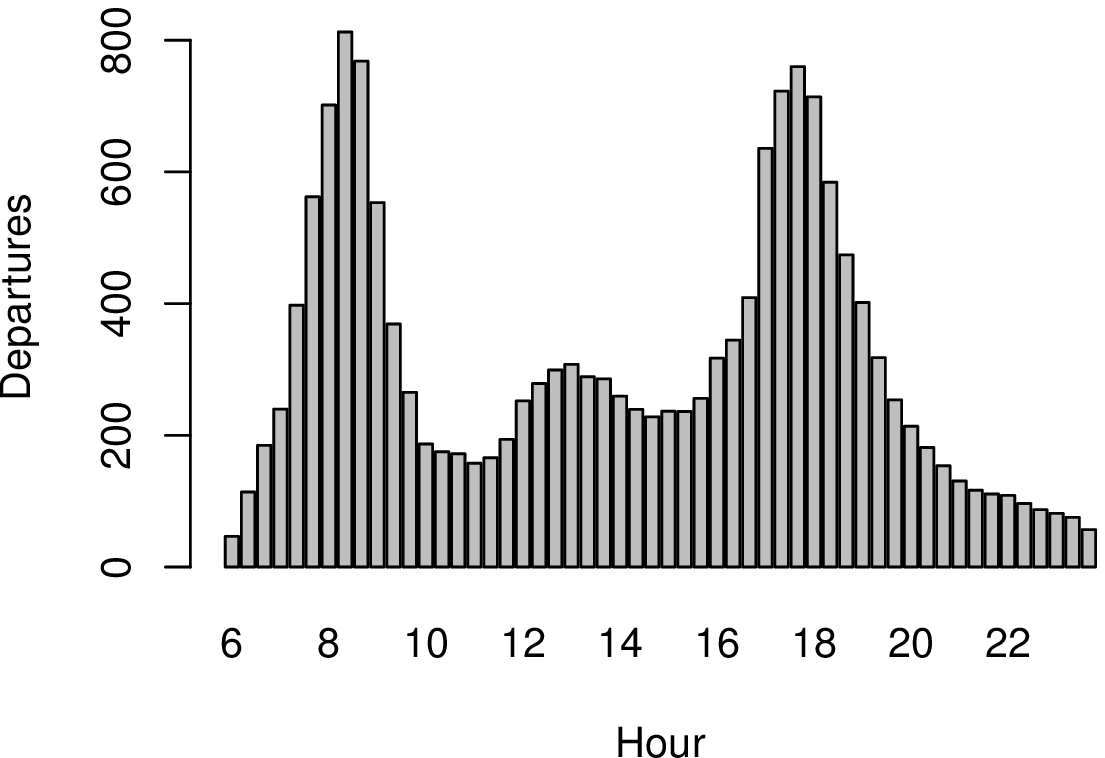}
      \caption{Average departure rates on weekdays}
      \label{fig:rides-workday}
    \end{minipage}
    \hfill
    \begin{minipage}[h]{.475\linewidth}
      \centering
      \includegraphics[width=\textwidth]{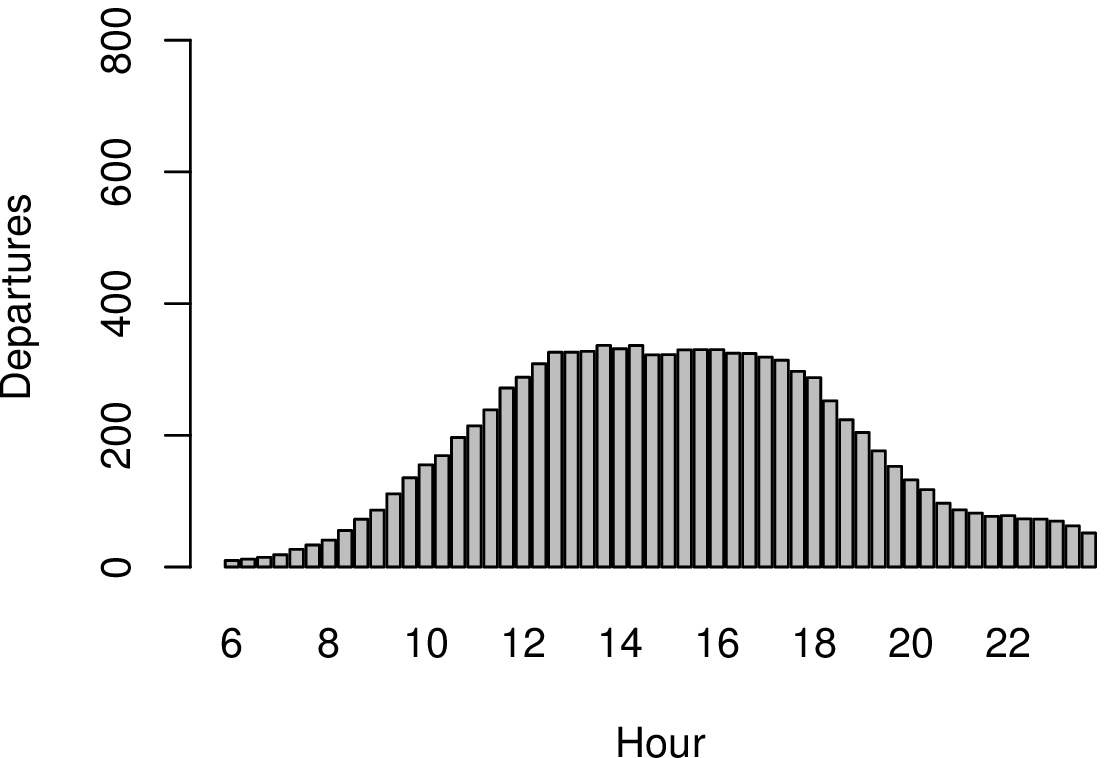}
      \caption{Average departure rates on weekends}
      \label{fig:rides-weekend}
    \end{minipage}
    \label{fig:daily-rides}
  \end{figure}
\end{subfigures}

\begin{subfigures}
  \begin{figure}
    \begin{minipage}[h]{.475\linewidth}
      \centering
      \includegraphics[width=\textwidth]{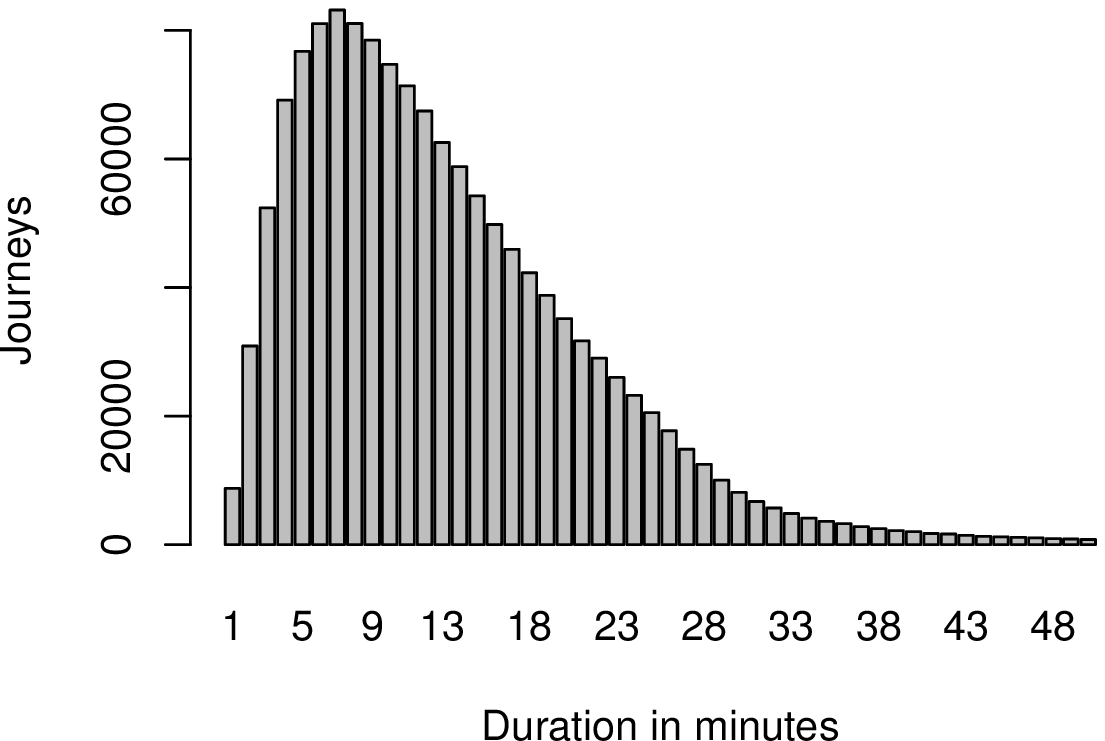}
      \caption{Distribution of the journey duration in the historical data}
      \label{fig:journey-duration}
    \end{minipage}
    \hfill
    \begin{minipage}[h]{.475\linewidth}
      \centering
      \includegraphics[width=\textwidth]{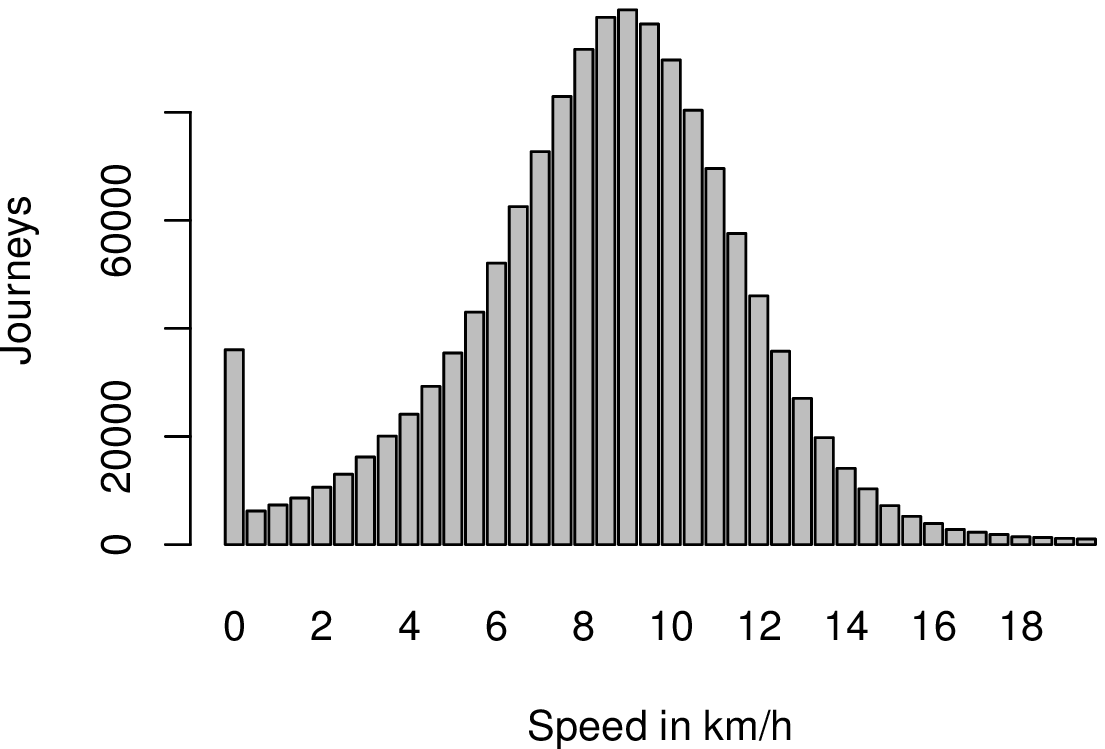}
      \caption{Journey speed distribution (Euclidean station distances)}
      \label{fig:journey-speed}
    \end{minipage}
    \label{fig:journey-figures}
  \end{figure}
\end{subfigures}

\subsection{Model parameters}

In our model, we define a set $S$ containing all stations $s\in S$. Time $t$ is
assumed to be discrete and indexed on a one-minute level where $\Thist$ denotes all
time steps of the observed period. We distinguish between workdays and days on
weekends by the binary variable $w\in\{weekday, weekend\}$ and split every day
into 72 slices $k \in K$ of 20 minutes each.
Time is mapped to day-type and timeslice using $w(t)$ and $k(t)$, respectively.
All customer departure and arrival events are counted in matrices of dimension
$|S|\times |S|$. The sum of departing customers going from station $i$ to $j$ in
a timeslice $k$ and on a day $w$ is $D_{i,j}(k,w)$; similarly the sum of
customers who arrive at station $j$ coming from $i$ is $A_{i,j}(k,w)$. The
average number of such arrival events $\Lambda_{i,j}(t)$ and departure events
$M_{i,j}(t)$ at time $t$ in the historical data can thus be expressed as
\begin{align}
  M_{i,j}(t) &= \frac{D_{i,j}\left(k(t),w(t)\right)}{|\{t'\in \Thist: k(t') = k(t), w(t') = w(t)\}|},\\
  \label{eq:averagetotal}
  \Lambda_{i,j}(t) &= \frac{A_{i,j}\left(k(t),w(t)\right)}{|\{t'\in \Thist: k(t') = k(t), w(t') = w(t)\}|}.
\end{align}
Here, the denominator gives the duration of the recorded history for a given
timeslice and day-type indicated by $t$.

Based on these average numbers of events happening per time step, customer
departures are exponentially distribution with time-varying parameter
$M_{i,j}(t)$. This parameter fit is based on the following implicit assumptions:
\begin{itemize}
\item 100\% service rate for departures in the historical data. Potential
  customers who could not rent a bicycle due to an empty station are excluded,
  as they are not recorded in the historical data. This assumptions is justified
  to some degree by the considerable repositioning effort made by the operator
  of the London PBS scheme \cite{stannard2011best}.
\item Independence of customer arrival. Departure of customers vary with time
  and type of the day, but do not depend on other departures. As a caveat, this
  does not accurately model customer groups, for example tourists.
\item Effects of season, weather, events, etc. are disregarded, but could easily be
  included in a more detailed model.
\end{itemize}
If a customer has departed at a station, the probability distribution of his
destinations is given by their relative frequency in the historical data, as
recorded in $\Lambda_{i,j}(t)$.

The total expected departure $\mu_s(t)$ and arrival $\lambda_s(t)$ at each station and
the \emph{net change of fill level} $\eta_s(t)$ during time step $t$ is therefore
\begin{equation}
  \mu_s(t) = \sum_{\tilde s\in S} M_{s, \tilde s}(t),\hspace{1em} \lambda_s(t) = \sum_{\tilde s\in S}\Lambda_{s,\tilde s}(t),
\end{equation}
\begin{equation}
  \label{eq:eta}
  \eta_s(t) = \lambda_s(t) - \mu_s(t).
\end{equation}
In order to simulate the system, the following assumptions about the behavior of
customers are needed, in addition to their arrival rates:
\begin{itemize}
\item Customers who want to depart from a station that turns out to be empty leave without starting a journey. They do not wait for a bicycle to be
  returned, nor do they walk on to a neighboring station.
\item The travel time between any two stations $i$ and $j$ is always equal to
  the average travel time extracted from the historical data.
  Figure~\ref{fig:journey-figures} depicts the historical distribution of travel
  speeds and journey durations.
\item Customers who arrive at their target station wanting to return their
  bicycle when the station is full ride on to one of the neighboring stations
  (chosen according to his perceived utility, as described in Section
  \ref{sec:cust-react-incent}). If this station is also full, a customer will go
  on to the next station, but he does not return to any station already visited.
\end{itemize}

\subsection{Customer decision model}
\label{sec:cust-react-incent}
In order to investigate the effects of offering price incentives, a model of how
the customer reacts is required. We assume all customers place a value on the
additional time they would spend travelling if they were to accept an incentive.
This is equivalent to penalizing a longer travel distance. The additional
distance a customer has to travel if he changes his target station consists of
the additional distance he has to bike, plus an additional walking distance to
his final destination. We assume this final destination lies at the center of
mass of the Voronoi region around each station
(Figure~\ref{fig:voronoi-system}). The Voronoi region is the polytope that contains
all points closer to a given station than to any other.
\begin{figure}[ht]
    \centering
    \includegraphics[width=.6\textwidth]{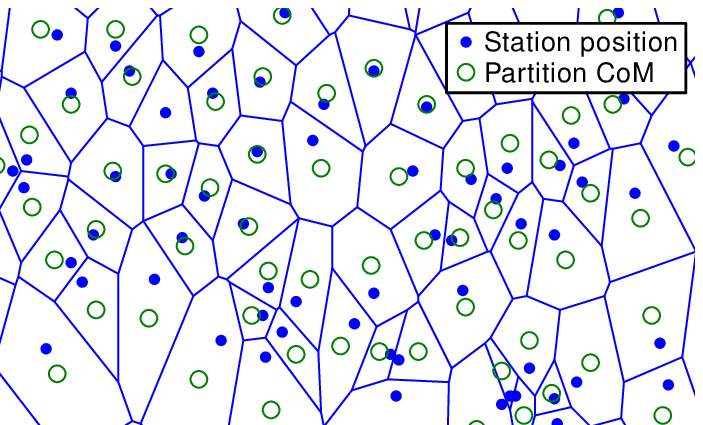}
    \caption{Voronoi partitioning with centers of mass (CoM) of the
      London bike sharing system}
    \label{fig:voronoi-system}
\end{figure}

Let the center of mass of the Voronoi region around each station $s_{i}$ be denoted
by $m_{i}$, and let $d_{\textrm{eucl}}$ be the Euclidean distance on the map.
Assuming that the walking speed is half the cycling speed, the \emph{effective
  distance} $\tilde d_{i,j}$ between two stations $s_{i}$ and $s_{j}$ can be
expressed as
\begin{equation}
  \label{eq:combined-distance}
  \tilde d_{i,j} = d_{\textrm{eucl}}(s_i,s_j)+2\cdot d_{\textrm{eucl}}(s_j,m_j)-2\cdot
  d_{\textrm{eucl}}(s_i,m_i).
\end{equation}

The incentives to go to a neighboring station are offered to customers upon
their arrival. Each customer decides whether to take an incentive by maximizing
his personal benefit based on the incentive payout and the customer's perceived
cost of additionally traveled distance. This implies customer rationality and
makes the choice independent from the original pricing of journeys. For each end
station $s\in S$, the set of neighboring stations to which a 
price incentive could be offered is $N_s$, and $p_{s,n}$ denotes the amount of money
offered to from station $s$ to neighbor $n\in N_s$. In addition, let $\tilde
N_s$ be the set of stations $\tilde s$ which have $s$ in their neighbor set. The
following model of customer reactions is used.
\begin{enumerate}
\item The marginal cost of travel $c$ for each arriving customer is drawn from a
  uniform distribution $C\sim U[0,c_{\max}]$, where we have used
  $c_{\max}=\pounds 20$/km in our simulations.\footnote{Future work could
    incorporate more detailed models of how customers value
    their time (see \cite{mackie2001value, hess2005estimation}). In addition, instead
    of having a single distribution, one could also differentiate between
    customer types (e.g. those commuting to work and those riding during leisure
    times) by introducing a time-varying component.}
\item The customer selects the best offer of maximum value as
  \begin{equation}
    \label{eq:bestvalue}
    n^* = \underset{n\in N_s}{\arg\,\max} p_{s,n}-d_{s,n}c.
\end{equation}
\item If the original target station is full, i.e.\ the customer cannot return
  his bike there, he always chooses the best incentive to go to $n^*$. If there
  is space, he takes an incentive only if the perceived value of the best
  incentive is positive, i.e.\ if $p_{s,n^*} - d_{s,n^*}c > 0$.
\end{enumerate}
The probability $\pi(s,n,p_s)$ of an arriving customer taking an incentive to
neighbor $n\in N_s$ for a given payout vector $p$ depends on the distribution of
perceived travel costs $c$. First, the offering to go to $n$ must have the
highest perceived value amongst all incentives offered to neighboring stations.
Second, assuming the station $s$ is not full, the perceived cost for traveling
the additional distance must be lower than the relevant payout $p_{s,n}$.
\begin{equation}
  \label{eq:prob-incentives}
  \begin{aligned}
  \pi(s,n,p_s) &= P\left(p_{s,n} \geq c\cdot\tilde d_{s,n}\ \land\ p_{s,n}-(c\cdot \tilde d_{s,n})\right.\\
&\left.\geq p_{s, n'}-(c\cdot \tilde d_{s,n'}),\forall n'\in N_s\right)
  \end{aligned}
\end{equation}

\section{Utility of changes in station fill level}
\label{sec:utility}
In this paper, two methods are considered for influencing the distribution
bicycles in the PBS: manual repositioning (Section \ref{sec:trucks}) and
price-led repositioning via incentives (Section \ref{sec:pricing}). However, as
a basis for both algorithms, it is necessary to estimate any change in the
stations' fill levels will bring about. Since the system is stochastic, it is
not straightforward to assess these benefits. In this section we introduce a
function that estimates the utility of changes in fill levels for a given
station.

Raviv and Kolka \cite{raviv_optimal_2013} have done related work in order to
determine the best fill level of each station in a static repositioning setting.
Their approach tracks the probability of all possible fill levels based on a
discrete approximation of the underlying continuous birth-death process.
However, if we were to adopt this method, the dimension of the resulting
optimization problem would significantly increase the computational complexity
of the proposed approaches in Sections~\ref{sec:trucks} and~\ref{sec:pricing}.
Therefore we propose a simpler approach.

We make the simplifying assumption that arrivals and departures are
deterministic and given by the expected net change $\eta_s(t)$. Furthermore, we
define the utility of changes to a station's fill level as the difference in the
number of customers expected to be served successfully at that station within a
long enough (but finite) time horizon. The benefit of any repositioning action
(adding or taking away bikes at a single station) at the current time can then
be evaluated based on this notion.

For a given station $s$ and starting time $t_0$, we precompute the expected
future fill level $f^s_t$ over a time horizon with $t = t_0, \ldots,
t_0+\Tutil$, where $\Tutil = 24$ hours is the look-ahead period considered. The
expected fill level is governed by the following dynamics:
\begin{equation}
\label{eq:fillupdate}
f^s_{t+1} = \max\left(0, \min(f^s_t+\eta_s(t),f^s_{\max})\right),
\end{equation}
where $f^s_{\max}$ denotes the maximum capacity of the station, and the $\max$
and $\min$ functions ensure that the station never becomes ``more than full'' or
``less than empty''. The quantity $\eta_s(t)$ is the net arrival rate defined by
(\ref{eq:eta}).

Adding or taking away bikes from the station at the current time changes how
many customers can be served later on in the time horizon, since the station
will become empty or full at different times in the future. For a current fill
level $f^s_{t_0}$ and a change in fill level $\Delta f$ to be made at $t_0$,
Algorithm~\ref{alg:utility} computes the utility $u(s,t,f^s_t,\Delta f)$ by
comparing the two cases of different initial fill level. The algorithm moves
through the time horizon forward in time, and in each time step compares the
amount of change $\Delta$ and $\tilde \Delta$ resulting from the system dynamics
and the station size constraints. If the amount of change is lower in the
original case, a station size constraint is being hit earlier than in the
adapted case and vice versa. The difference between $\Delta$ and $\tilde \Delta$
is the difference in customers that could be served successfully in that time
period. The procedure aborts if the end of the time horizon has been
reached or if the fill level of the station becomes equal in both cases (e.g.
both are full or empty).
\begin{algorithm}
\begin{algorithmic}[1]
\Require $s\in S, \Tutil$ \Comment{Relevant station and horizon length}, \\
$f^s_{\max}$ \Comment{Maximum capacity of station $s$}, \\
$\eta_s(t)$ \Comment{Expected net arrival of customers}, \\
$f^s_t\ \forall t\in \{t_0, \ldots, t_0+\Tutil\}$ \Comment{Precomputed fill levels in the original case according to (\ref{eq:fillupdate})}, \\
$\tilde f^s_{t_0}$ \Comment{Starting fill level in the case with repositioning}
\Procedure{RepositioningUtility}{$t,\tilde f^s_t$}
\If{$\tilde f^s_t = f^s_t$ \textbf{or} $t \geq \Tutil$}
\State \Return 0
\EndIf
\State $\tilde f^s_{t+1} \leftarrow \max(0, \min(\tilde f^s_t+\eta_s(t),f^s_{\max}))$
\State $\Delta \leftarrow |f^s_t-f^s_{t+1}|$
\State $\tilde \Delta \leftarrow |\tilde f^s_t - \tilde f^s_{t+1}|$
\State \Return{$\tilde \Delta- \Delta + $\Call{RepositioningUtility}{$t+1,\tilde f^s_{t+1}$}}
\EndProcedure
\end{algorithmic}
\caption{Computing the utility of repositioning}
\label{alg:utility}
\end{algorithm}

Although the worst-case time complexity of Algorithm \ref{alg:utility} is linear
with respect to the horizon length, it would take too long to use online in the
later optimization steps. Storing the results in a lookup table is intractable,
especially if the expected future fill levels are treated as continuous values.
However, fast computation can be achieved by constructing a simpler function,
making use of a lookup table of dimension $2\times|S|\times \Tutil$, from which
this utility can be determined. Figure~\ref{fig:plateau} shows an example of
such a utility function for an empty station, $u(s,t,f^s_t=0,\Delta f)$. For
simplicity, $\Delta f$ is also relaxed to be non-integral. The validity of this
parameterisation is now proven.
\begin{figure}
  \centering
  \includegraphics[width=.6\textwidth]{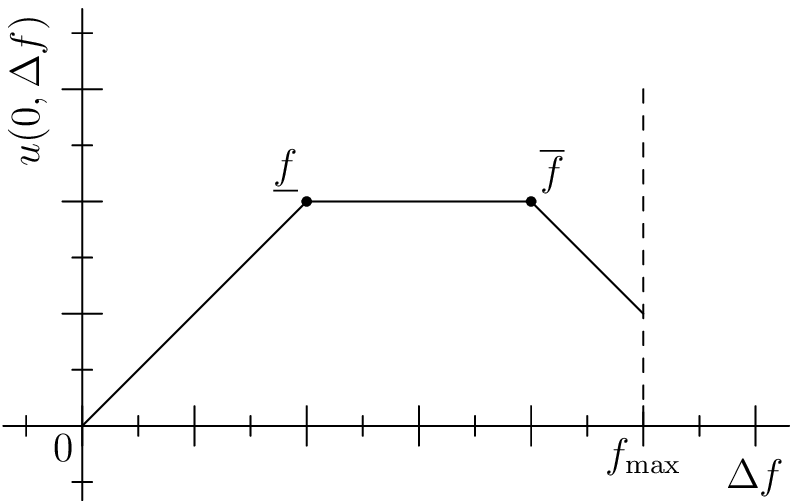}
  \caption{Example for an empty station ($f=0$), showing the utility ``plateau''
    between the fill levels $\underline f^s_t$ and $\overline f^s_t$. In this case maximum utility results from adding 4 to 8 bicycles at $t_0$, based on the expected net arrival rate over the time horizon.}
  \label{fig:plateau}
\end{figure}

\begin{theorem}
  \label{eq:theorem1}
  For any station $s\in S$ at time $t$, there exist two fill levels $\underline f^s_t,
  \overline f^s_t \in [0, f^s_{\max}]$ independent from the initial fill level
  $f^s_t$ such that for the utility of change in fill level $\Delta f$
  \begin{equation}
    \frac{\partial u(s,t,f^s_t,\Delta f)}{\partial\Delta f} = \begin{cases}
      1, &\text{if } f^s_t+\Delta f < \underline f^s_t \\
      0, & \text{if } \underline f^s_t \leq f^s_t+\Delta f \leq \overline f^s_t\\
      -1, & \text{if } f^s_t+\Delta f > \overline f^s_t.
    \end{cases}
    \label{eq:udiff}
  \end{equation}
  The utility of no change of the station's fill level is understood to be zero,
  $u(s,t,f^s_t,0) := 0$. The (possibly empty) interval $[\underline f^s_t,
  \overline f^s_t]$ of utility is called the ``plateau'' of constant maximum
  utility.
\end{theorem}
\begin{proof}
  Assume the station will become full within the time horizon $\Tutil$ for
  initial fill level $f^s_t+\Delta f$. Adding $\delta$ bicycles to the station
  implies that $\delta$ additional expected customers who want to return their
  bicycles have to be rejected. Therefore, the utility $u(s,t,f^s_t,\Delta f+\delta)$
  decreases monotonically with slope -1 for any $\delta \geq 0, f+\Delta f+\delta \leq f^s_{\max}$.
  \[u(s,t,f^s_t,\Delta f+\delta) = u(s,t,f^s_t,\Delta f)-\delta\] A similar case
  can be made for fill levels where the station is running empty. Only, the
  utility decreases when more bicycles are removed with a slope of $u$ equal to
  1. It follows naturally that for a given initial fill level there are
  thresholds $\underline f^s_t, \overline f^s_t$, with $\underline f^s_t <
  \overline f^s_t$, where the station first starts to run empty or full within
  the horizon. Within the interval $[\underline f^s_t, \overline f^s_t]$, the
  station's capacity constraints are not hit and the utility function must
  therefore be constant, leading to a ``plateau'' of the type shown in
  Fig.~\ref{fig:plateau}.
\end{proof}

It can be shown that $u(s,t,f^s_t,\Delta f)$ can be computed for each station using
only three calls to Algorithm \ref{alg:utility}. Sections \ref{sec:trucks} and
\ref{sec:pricing} will make use of this characterization of station fill
utilities in order to choose how trucks reposition bikes and price incentives
can be offered to customers.

The system dynamics (\ref{eq:fillupdate}) are based on deterministic net arrival
of customers. This results in a coarser model than for example the probabilistic
approach of \cite{raviv_optimal_2013}. However, this simple parameterization is
attractive in that it leads to tractable optimization problems, as will be seen
in Sections \ref{sec:trucks} and \ref{sec:pricing}.

\section{A dynamic truck-routing algorithm}
\label{sec:trucks}

This section describes an algorithm for intelligent operation of a fleet of
$R$ trucks, which move bikes between stations as needed. Their objective
is to increase the system utility (as defined in Section \ref{sec:utility}), and
hence ultimately the system's service level, as much as possible.

The problem of manual relocation of vehicles in shared vehicle systems is not
new. It originated in pilot car-sharing projects, such as the French
\emph{Praxitèle} \cite{dror_redistribution_1998,duron_analysis_2000},
\emph{Intellishare} \cite{barth2004} and \emph{Honda ICVS}
\cite{kek_decision_2009}. However, the repositioning algorithms used in
car-sharing projects do not translate directly to the public bicycle-sharing
scheme considered in this paper. Firstly, each of these algorithms exhibits
certain characteristics that are specific to its corresponding vehicle-sharing
system, for example charging times of electric vehicles. Secondly, car-sharing systems
tend to be much smaller in their network size than the PBS considered in our
paper, and the proposed algorithms cannot easily be scaled to several hundred stations.

Intelligent repositioning in bicycle-sharing schemes has also received prior
attention in the literature \cite{shaheen_bikesharing_2010}. The proposed
approaches can be separated into static and dynamic approaches.

In the static repositioning approach, an optimal route is computed in order to
attain a predefined fill level for each station, prior to customers interacting
with the system (e.g.\ during the night). For example,
\cite{benchimol_balancing_2011} present a solution for the routing of a single
truck, and \cite{raviv_static_2012} consider the case of multiple trucks. The
advantage of this approach is that there is ample time to compute a good truck
routing solution, and this solution could serve as a reliable basis for
computing the price incentives (Section \ref{sec:pricing}). However, static
repositioning has shown too little flexibility to react to unforeseen variations
in the demand pattern, caused for example by unexpected weather conditions.

In the dynamic repositioning approach, the truck routing is planned in a
receding horizon fashion while the system is in full operation. This allows the
planning to react online to unexpected changes in the system's state. As such it
is a more suitable approach to our problem. Rair and Miller-Hooks
\cite{nair_fleet_2011} use a stochastic formulation for dynamic repositioning
based on stationary distributions for customer arrivals and departures. However,
their approach is not suitable to the case of our PBS, since these distributions
vary considerably according to the distinct daily flow patterns described in
Section \ref{sec:intro}. Contardo et al. \cite{contardo_balancing_2012} present
another dynamic repositioning approach with time-varying, yet deterministic
future flow patterns. But the computational complexity of their approach is
prohibitive for our system, because it is too expensive to simulate the system
with multiple trucks over a long time horizon.

Whilst many PBS schemes redistribute bicycles during the night, nighttime
operation is restricted in London in several key areas
\cite{redistribution_2012}. Therefore, this paper considers the dynamic
repositioning case only. It is a variation of the routing problem with pickups
and deliveries for one commodity, taken from \cite{berbeglia2010dynamic}. It is
illustrated in Figure \ref{fig:truckroute}. In order to determine a truck route,
time is discretized into 5-minute intervals and the planning problem is
considered on the time-expanded network (Section \ref{sec:stationtime}) During
each interval, customer behavior is assumed to be time-varying, but
deterministic. A receding planning horizon is considered, which is the maximum
of the truck visiting $4$ stations and $\Ttruck= 40\,\text{min}$. The period for
re-optimizing of the truck routes (``implementation horizon'') is chosen as
$T_{\textrm{impl}}=30\,\text{min}$, based on a trade-off between computation time and performance
quality. Note that as shown in Figure \ref{fig:truckroute}, the planning horizon
is longer than the re-optimization cycle. This improves the performance of truck
journeys beginning shortly before the next re-optimization as these journey are
likely to end within the longer planning horizon where subsequent opportunity is
still considered.

To solve the routing problem for a single truck (Section \ref{sec:greedy}) over
the finite planning horizon, we adopt a two-step approach. First, for each truck
we construct a tree of ``promising route candidates'' using a greedy heuristic.
This means that truck routes are extended by stations that promise high ratios
of utility added per time to travel. Second, for each of the promising routes,
the optimal number of bikes to be loaded or unloaded at each stop is optimized
as the complete routes have become known. Then the route providing the highest
utility improvement is selected. Finally, we extend this algorithm in a simple
manner to multiple trucks (Section \ref{sec:multiple}).

\begin{figure}[t]
    \centering
    \includegraphics{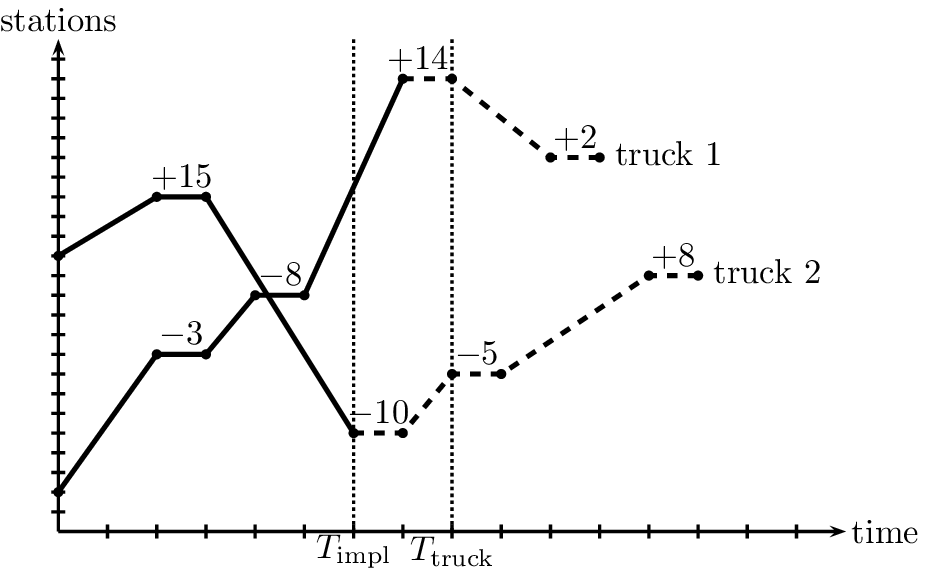}
    \caption{Illustration of the dynamic truck-routing algorithm for $R=2$
      trucks. The time axis is discretized into $5\,\text{min}$-intervals. The
      planning horizon is the maximum of $\Ntruck=4$ stations visited by the
      truck and $\Ttruck=40\,\text{min}$; the implementation horizon is
      $T_{\textrm{impl}}=30\,\text{min}$. Each line represents the route of one
      truck, where journeys are indicated by a solid line if they start within
      the implementation horizon (i.e.\ they are definitely executed) or by a
      dashed line if they start with the planning horizon (i.e.\ they may be
      subject to change at future re-plannings). The trucks wait at each stop
      for $5\,\text{min}$ in order to to load or unload bikes. The number of
      bikes loaded and unloaded at each stop is indicated as well.
      \label{fig:truckroute}}
\end{figure}

\subsection{Modeling repositioning truck routes on a time-expanded network}
\label{sec:stationtime}
We now describe how the truck routes can be modeled as a time-expanded network
on a graph $G=(V,A)$ \cite{kek_decision_2009,contardo_balancing_2012}, which is
the basis of our truck-routing algorithm. The vertices $V$ of the graph consist
of tuples $v=\{(s,t),s\in S,t\in T\}\in V$.\footnote{The notation $s(v)$ is used
  to access the component $s$ of the tuple $v=(s,t)$.} The (expected future)
fill level $f$ for each vertex $v$ (i.e.\ station $s$ at time $t$) is generated
according to (\ref{eq:fillupdate}). The arcs $a = (v_1, v_2) \in A$ of $G$
correspond to possible journeys a truck is able to take.

In our model, the time it takes for a truck to traverse an arc is discretized to
multiples of 5 minutes. It is computed based on the Euclidean distance (in km)
$d_{\textrm{eucl}}(s_{i},s_{j})$ between two stations $s_{i},s_{j}\in S$,
assuming an average speed of $15\,\text{km/h}$ for the truck in city traffic.
Including an additional $5\,\text{min}$ for bicycle handling after reaching the
station, the resulting \emph{effective journey time} (in time steps of 5 min) for a truck to go from
station $s_{i}$ to $s_{j}$ is
\begin{equation}
  \bar d(s_{i},s_{j}) := \lceil d_{\textrm{eucl}}(s_{i},s_{j})/1.25 \rceil + 1.
\end{equation}
Note that the dividing factor of 1.25 results from converting distance (in km)
into time steps (of 5 min). The repositioning trucks have limited operation
hours during the day, set to $7\,\text{am}-10\,\text{pm}$. All trucks are
constrained to start at a maintenance depot in the morning, and also to finish
at this depot at the end of the working day. As a consequence, vertices from
which no combination of arcs leads back to the depot on time are excluded from
the graph. An example of a network of stations and a corresponding time-expanded
network are shown in Figures~\ref{fig:station-graph}
and~\ref{fig:station-time-graph}.
\begin{figure}
\begin{minipage}[b]{0.45\linewidth}
\centering
\includegraphics[height=4cm]{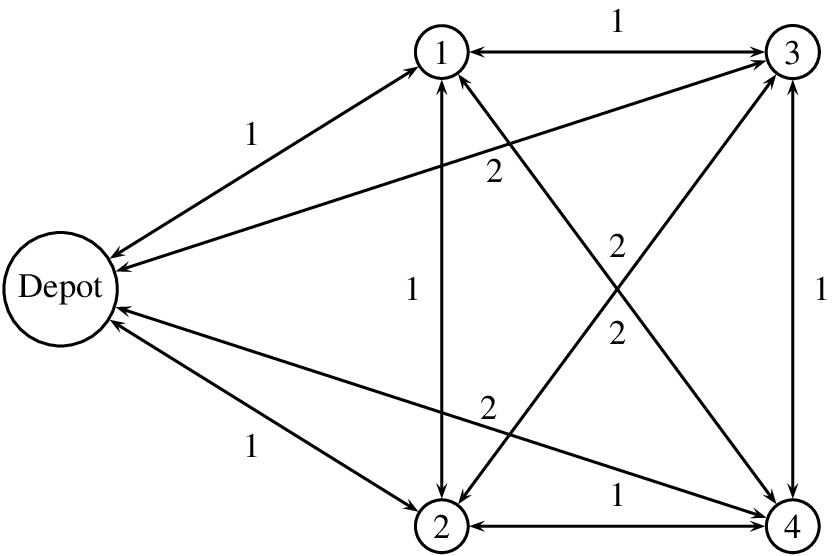}
\caption{Example station graph, with arcs weighted according to distance.}
\label{fig:station-graph}
\end{minipage}
\hfill
\begin{minipage}[b]{0.53\linewidth}
\centering
\includegraphics[width=\textwidth]{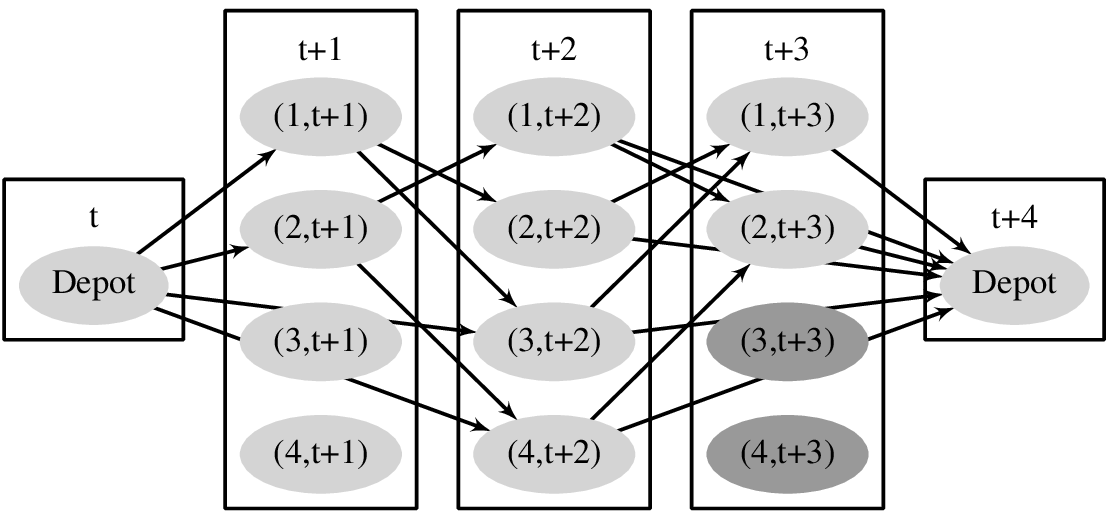}
\caption{Example time-expanded network. Dark-grey vertices are marked ``dead''
  due to the terminal condition.}
\label{fig:station-time-graph}
\end{minipage}
\end{figure}

\subsection{Computing single-truck routes}\label{sec:greedy}
In this Section, we present an algorithm for finding a repositioning route for a
single truck in the dynamic case. The amount of time available for repositioning
is assumed to be fixed. In contrast to the static repositioning case, the goal
is not to attain a defined system state as little resources as possible, but
rather to optimally invest the available resources, i.e. operational hours of
the trucks. Thus, based on the utility of fill level changes defined in
Section~\ref{sec:utility}, the ratio of added utility per invested time of a
truck $r$ is maximized.

\paragraph{Constructing a tree of promising candidate routes}
Every truck $r\in R$ can hold a maximum load $l_{\max}=20\,\text{bikes}$ and holds
$l^r_t\in\{0, \ldots, l_{\max}\}$ bicycles at time $t$. The goal is to determine
a truck's \emph{repositioning actions} $\rho_i = (v_i, \Delta f_i, l^r_i)\in P:
v_i\in V, l_i\in\{0, \ldots, l_{\max}\}, i\in\{1, \ldots, \Ntruck\}$. The truck
load $l^r_i$ is defined as the number of bicycles after the repositioning action
$\Delta f_i$ performed at the station and time indicated by the vertex of the
time-expanded network $v(\rho_i)$. In particular, for every pair of actions
$\rho_i, \rho_{i+1}$ there must exist an arc $a\in A$ with
\begin{equation}
v_1(a) = v_i\ \land\ v_2(a) = v_{i+1},\quad  \forall i\in\{1, \ldots, \Ntruck-1\}.
\end{equation}
Moreover, the following consistency constraints for station fill levels
and truck loads must hold:
\begin{subequations}
  \begin{equation}
    f^{s(v_i)}_{t(v_i)} = f^{s(v_i)}_{t(v_i)-1} + \Delta f_i \in [0, f^{s(v_i)}_{\max}],\hfill \forall i\in \{1,\ldots, \Ntruck\}
  \end{equation}
  \begin{equation}
    l^r_{i+1} = l^r_{i}- \Delta f_{i+1}\in \{0, \ldots, l_{\max}\},\hfill \forall i \in \{1,\ldots, \Ntruck-1\}
  \end{equation}
\end{subequations}

Starting from an initial repositioning position $\rho_1$, the possible next steps
can be represented by a tree graph $\Phi$, where each node $\phi$ represents a
specific repositioning action and each leaf node determines a unique truck
route. The tree of all possible routes has a branching factor of $|S-1|$, since a
route could possibly lead to any of the other stations for the next
repositioning action. So there are $|S-1|^{\Ntruck-1}$ possible combinations of
stations for a route of length $\Ntruck$ (where the initial position $\rho_1$ is
already known). For systems consisting of several hundred stations it is thus
not viable to test every possible combination. The complexity of the problem is
reduced by concentrating on a subset of possible routes corresponding to the
most promising repositioning actions. It works by constructing a pruned version
of the route tree. Starting from the trucks initial position (at the root), the
tree is recursively extended at each of its leaf nodes until it has reached the
desired height of $\Ntruck$ (or the journey has a minimum duration of
$\Ttruck$):

\begin{enumerate}
\item The current leaf node is $\phi = (v_{\phi},\Delta f_{\phi}, l_{\phi})$. First, we compute the
  ``value per unit distance'' the truck might bring by going to any of the other
  stations. The set $\tilde V:=\{v\in V:\exists a\in A, v_1(a) = v_{\phi},
  v_2(a) = v\}$ contains the vertices of the time-expanded network the truck
  would reach by going to any of the other stations. Since we know the current
  load of bicycles $l^r_{\phi}$, the best action to be performed at $\tilde v\in
  \tilde V$ can be computed with a ``greedy'' approach as
\begin{equation}
\begin{aligned}
  \label{eq:best-action}
    &\Delta f^*(\tilde v,\phi) =\\
 &\begin{cases}
      \max\left(l^r_\phi-l_{\max}, \left\lceil \overline f^{s(\tilde v)}_{t(\tilde v)} - f(\tilde v)\right\rceil\right), &\text{if } f(\tilde v) > \overline f(\tilde v)\\
      \min\left(l^r_{\phi},\left\lfloor \underline f^{s(\tilde v)}_{t(\tilde v)} - f(\tilde v)\right\rfloor\right), &\text{if } f(\tilde v) < \underline f(\tilde v)\\
  0, &\text{else}.
\end{cases}
\end{aligned}
\end{equation}
We choose the $K$ vertices with the best ratio $\Delta f^*(\tilde v,\phi) / \bar
d(v_{\phi},\tilde v)$ and add them as leaves of $\phi$ in the form of route
steps. The set of new leaf nodes is $\Phi_\phi$.
\item In addition, we add stations that could serve as an intermittent depot.
  Going there may not yield a direct utility. But the possibility to bring or
  take bicycles may be of use at other stations of the route. We choose
  \begin{IEEEeqnarray}{rCl}
      \tilde v_{store} &=& \underset{\tilde v\in \tilde V: s(\tilde v)\neq s(v_{\phi})}{\arg\,\max} \frac{\min\left(l_{\textrm{depot}},\overline f^{s(\tilde v)}_{t(\tilde v)}-f(\tilde v)\right)}{\bar d\left(v_{\phi},\tilde v\right)}\IEEEeqnarraynumspace\IEEEyessubnumber\\
      \tilde v_{pick} &=& \underset{\tilde v\in \tilde V: s(\tilde v) \neq s(v_{\phi})}{\arg\,\max} \frac{\min\left(l_{\textrm{depot}},f(\tilde v)-\underline f^{s(\tilde v)}_{t(\tilde v)}\right)}{\bar d\left(v_{\phi},\tilde v\right)}\IEEEyessubnumber
  \end{IEEEeqnarray}
  and add them to the set of leafs $\Phi_\phi$ with a repositioning action of zero. To
  prevent that $\tilde v_{store}, \tilde v_{pick}$ are only set to very large
  stations, we cap the maximum intermittent depot size considered to some
  $l_{\textrm{depot}}\leq l_{\max}$. How stations that may serve as intermittent
  depots are incorporated into the actions at other steps of the route is
  explained in the following.
\item If the depth of the recursive procedure has not yet reached the final
  depth $\Ntruck$, it is repeated for every $\phi'\in \Phi_\phi$. Before
  entering the recursive procedure at $\phi'$, the corresponding repositioning
  action $\Delta f_{\phi'}$ is incorporated into the predicted future fill levels
  of the time-expanded network. These changes, of course, have to be unwound
  between the $\phi' \in \Phi_\phi$.
\end{enumerate}

If even more aggressive tree pruning is necessary to comply with computational
constraints, the similar Beam Search \cite{lowere1976harpy, peng1988filtered}
can be applied. It leads to linear complexity in the route length, but at the
expense of discarding more potentially optimal solutions. In Beam Search, a
greedy approach is used to determine promising next steps as well. But only $K$
leaf nodes are added in total to all nodes of the same height. Resorting to Beam
Search was, however, not necessary for the route length horizon used in the sample
setting of this paper.

\paragraph{Refining truck loading actions}
The repositioning actions $\rho_i,\ i=1,\ldots,\Ntruck$ of the promising routes
candidates in $\Phi$ stem from a greedy heuristic that could not know about
stations visited later in the route. Knowing the complete routes, their
respective action profile should be further refined. As a motivating example,
it may be beneficial to pick up more bicycles than the utility function of a
single station $u(f,\Delta f)$ (see Section~\ref{sec:utility}) originally
indicated. That is, if taking more bicycles has zero utility locally (the
fill level remains within the utility plateau), but a subsequent stations in the
route can make use of the additional bicycles. The problem of choosing optimal
actions can be formulated as a manageable quadratic program (QP).
\begin{IEEEdescription}[\IEEEsetlabelwidth{$\Delta \underline f'(i),\Delta \overline f'(i)$}\IEEEusemathlabelsep]
\item[$s(i), t(i)$] The station and time at the $i$-th step in the route for $i\in
  \{1,\ldots, \Ntruck\}$.
\item[$f_i$] The expected fill level of station $s(i)$ at time $t(i)$.
\item[$\Delta f_i$] The action performed at step $i$. This is the optimization
  variable.
\item[$\underline f(i),\overline f(i)$] Beginning and end of the utility plateau
  of station $s(i)$, as described in Section~\ref{sec:utility}.
\item[$\Delta \underline f(i),\Delta \overline f(i)$] Difference between the new fill
  level $f_i + \Delta f_i$ and the plateau beginning/end $\overline
  f^{s(i)}_{t(i)}, \underline f^{s(i)}_{t(i)}$. The difference is defined to
  grow positively going outwards from the respective side of the plateau.
\item[$\Delta \underline f'(i),\Delta \overline f'(i)$] Auxiliary variables containing the
  absolute difference from the plateau beginning $\Delta \underline f'(i) =
  |\Delta \underline f(i)|$ or end $\Delta \overline f'(i) = |\Delta \overline
  f(i)|$. They are correctly set by the solver minimizing costs within the
  bounds set in (\ref{eq:absequ1}) and (\ref{eq:absequ2}).
\item[$l_{0}, l_{\max}$] Starting fill level of the repositioning truck and the
  maximum truck load capacity.
\item[$q \gg 2l_{\max}^2+1$] Scaling factor for penalizing repositioning actions.
\end{IEEEdescription}
\begin{equation}
  \min \sum_{i=1}^{\Ntruck} \Delta \underline f(i) + \Delta \underline f'(i) + \Delta \overline f(i) + \Delta \overline f'(i) + \sum_{i=1}^{\Ntruck} \Delta f_i^2/q \label{eq:optaction}
\end{equation}
such that
\begin{IEEEeqnarray}{rcl'l}
0 \leq l^r_{t_0} &-& \sum_{i'=1}^{i} \Delta f_{i'} \leq l_{\max}, & \forall i\in \{1, \ldots, \Ntruck\}\IEEEyessubnumber\label{eq:capacityconstraint}\\
\Delta \underline f(i)\;&=&\;\underline f(i) - f_i - \Delta f_i, & \forall i\in \{1, \ldots, \Ntruck\}\IEEEyessubnumber\\
\Delta \overline f(i)\; &=&\; -\overline f(i) + f_i + \Delta f_i,  & \forall i\in \{1, \ldots, \Ntruck\}\IEEEyessubnumber\IEEEeqnarraynumspace\\
-\Delta \underline f'(i)\; &\leq&\; \Delta \underline f(i) \leq \Delta \underline f'(i), & \forall i\in \{1, \ldots, \Ntruck\}\IEEEyessubnumber\label{eq:absequ1}\\
-\Delta \overline f'(i)\; &\leq&\; \Delta \overline f(i) \leq \Delta \overline f'(i), & \forall i\in \{1, \ldots, \Ntruck\}\IEEEyessubnumber\label{eq:absequ2}
\end{IEEEeqnarray}

The linear part of the objective function (\ref{eq:optaction}) evaluates
repositioning actions according to the utility definition of
Section~\ref{sec:utility}. The quadratic part of (\ref{eq:optaction}) minimizes
the action of the truck operators (they take the fewest bicycles possible). It
is scaled such that actions with a positive utility will still be performed;
However, it prevent actions causing negative utility at one station and the
equal positive utility at another step in the route. This also ensures that
stations will not be pushed outwards from the plateau and actions remain
feasible. Equation~(\ref{eq:capacityconstraint}) ensures that the fill level of
trucks stays within the capacity constraints. If computation time allows, an
integer constraint $\Delta f_n \in \mathbb{Z},\ \forall n$ can be added to
reflect the discrete number of bikes. This renders the optimization problem into
a mixed-integer quadratic program (MIQP). In our approach, we manually fit the
solution to the truck load constraints based on the relaxed QP solution by
clipping any non-integer parts. In test runs no or only very little differences
from the MIQP were observed.


We now determine the best set of repositioning actions for all promising routes
and choose the route that results in the best overall utility increase per unit
time.

\subsection{Routing multiple trucks}\label{sec:multiple}
Co-optimized routes for several trucks are too difficult to compute online
within the time constraints of the PBS system. Therefore, we resort to
optimizing multiple truck routes sequentially where actions of prior trucks can
be treated as known. These actions are manifest in the future fill levels stored
in the graph of the time-expanded network.

The basic idea is to repeat adding route steps for each truck until it reaches a
minimum number of $\Ntruck$ steps, or a journey time of $\Ttruck$. This is
performed sequentially, starting with the truck who has the minimum time-index
for the last step in his route and has not yet reached the required route
length. The reason for this sequential procedure is to prevent the collision of
two truck routes, which can be detrimental to the performance of the algorithm
as explained below. Assume, without loss of generality, that the trucks $r\in R$
are ordered according to the current computation of their routes. If truck $r$
chooses to go to a station $s(v)$ to which a truck $r'<r$ has already planned to
go at a later time $t(v')>t(v), s(v') = s(v)$, then $r'$ has made his choice
based on false assumptions about the station's fill level. These collisions can
be handled by
\begin{itemize}
\item Removing all routing steps that were added during the last route-search
  step from $\Omega_{r'}$.
\item Removing all but the first routing steps that were added during the last
  route-search step from $\Omega_r$.
\end{itemize}

So collisions are prevented and since at least one new routing step is preserved
per detected collision, so our algorithm will eventually reach $\Ttruck$ for all
trucks.

\section{Dynamic price incentives for users}
\label{sec:pricing}

Customers themselves might contribute to the rebalancing of a PBS scheme if
offered an appropriate payment. In this paper we consider how payments could be
offered to customers to change the endpoint of their journey to a nearby station
in a way that improves the overall service level. To this end, we take the model
of how customers accept (or choose between) price offers, as described in
Section \ref{sec:cust-react-incent}, and then form an optimization problem
trading off the expected payouts and the expected improvement in service level.
The solution of this optimization problem is a set of price offers that are presented
to any customer arriving at a given station. It seems reasonable to reduce the
complexity of this optimization problem by limiting the number of prices offered
(= decision variables) to 10 per station; i.e.\ for each station, only 10 prices
are quoted for going to selected neighboring stations.

We assume that means of communicating the price incentives and for making
payments are available. A payment infrastructure is already central to existing
systems, like the \emph{Oyster Card} for London's public transport network. New
information capabilities could be added to the kiosk terminals used for rental,
and/or the mobile applications many customers already use.

Using price incentives to induce a desired behavior in the users of shared
mobility systems has been examined in multiple contexts in recent literature.
The work of \cite{barth2004} examines user-based repositioning in a shared
mobility system. However, the approach of splitting and merging rides can only
be applied to cars and not to public bike hire schemes. Incentives for bike
sharing schemes are investigated by \cite{fricker_incentives_2012}, where users
pick two stations at random and go to the more empty.

In this paper, we propose a novel scheme that is based on Model Predictive
Control (MPC). A short summary of MPC is given in Section
\ref{sec:model-pred-contr}. In Section \ref{sec:simpl-model-cust} we then
show how the customer reaction model from Section \ref{sec:model} can be
linearized in order to obtain a tractable MPC problem formulation. Finally,
Section \ref{sec:mpc} explains the details of the corresponding
optimization problem that has to be solved in a
receding-horizon fashion in order to determine the real-time price incentives.

\subsection{Model Predictive Control}\label{sec:model-pred-contr}
Here we provide a short introduction to Model Predictive Control (MPC \cite{maciejowski2002}), giving
the basic explanations required to describe the controller developed
in Section~\ref{sec:mpc}. The fundamental idea is to employ a model of a given system in
order to optimize the inputs given to the system over a finite control horizon.
Only the first input is then applied to the system, and the scheme continues by
measuring the new state of the system and solving another finite-horizon
optimization problem (``MPC problem'').

The two main aspects of MPC comprise good control decisions for the system with
respect to anticipated future events (by the optimization), and feedback in the
case of unforeseen disturbances or model inaccuracies (by re-optimizing
periodically for the control actions). An important strength of MPC is its
ability to incorporate a model of the system dynamics and to handle constraints
on the states and control inputs. For example, in the case of the PBS this is
advantageous because it allows upper and lower bounds to be placed on the computed
price incentives (control inputs) and the stations' fill levels (system state).

Let the system state at time step $t$ (i.e.\ a vector containing the fill levels
of all stations) be denoted by $x(t)$. The future state evolves as some function
of the current state and the control inputs $u(t)$ (i.e.\ a vector of the price
incentives offered to the users), so that the subsequent state is given by
$x(t+1) = f_t(x(t), u(t))$. Here $f_t$ represents only a simplified model of the
actual system dynamics, which is subject to model uncertainty and disturbances.
Note that the function $f_t$ depends on the predicted customer interaction and
is therefore time-varying (recall that customer interaction shows time-varying
patterns). The actual deviations from the predicted customer interaction is
uncertain, and thus considered as a disturbance to our model. Moreover, given
that the truck routes are known (from Section \ref{sec:trucks}), the functions
$f_t$ also contain changes to the system state affected by manual repositioning.

Assume the current time to be $t=0$ without loss of generality. Now we wish to
choose a \emph{finite series} of inputs $u(t)$ where $t=0,\ldots,T-1$ so that
the system behaves optimally over the finite time horizon $T$, starting from the
measured current state $x(0)$. This is done by solving an optimization problem,
trading off the perceived cost of having suboptimal system states $c^x_t(x(t))$, and
the cost of applying the control input $c^u_t(u(t))$:
\begin{equation}
\label{eq:mpc1}
  \min_{u(0),\ldots,u(T-1)}\quad \sum_{t = 1}^{T} c^x_t(x(t)) + \sum_{t = 0}^{T-1} c^u_t(u(t))
\end{equation}
such that
\begin{IEEEeqnarray}{rCl'l}
x(t+1)&=& f_t(x(t), u(t)), & t = 0,\ldots,T-1\IEEEyessubnumber\\
u(t)&\in& \mathcal{U}_t, & t = 0,\ldots,T-1\IEEEyessubnumber\\
x(t)&\in&\mathcal{X}_t, & t = 1,\ldots,T\IEEEyessubnumber
\end{IEEEeqnarray}
where the functions $c^x_t$ and $c^u_t$ are called \emph{stage costs} for the
state and input respectively, and sets $\mathcal{X}_t$ and $\mathcal{U}_t$
represent any constraints that may be present on the state and input.

Although solving problem (\ref{eq:mpc1}) gives a series of inputs $u(t), t=0,
\ldots, T-1$, only $u(0)$ is applied to the system. In the next time step, a new
measurement of the current state is made, and a new series of planned control
inputs are determined by resolving the MPC problem in light of the new
information. For this reason, MPC is also known as ``Receding Horizon Optimal
Control''.

To make problem (\ref{eq:mpc1}) tractable, the system model must often be
simplified. In particular, non-linear dynamics $f_t(x(t),u(t))$ make the problem
non-convex. For many systems, though, good control performance can still be
achieved if the model is linearized.

\subsection{Simplified model of customer behavior}
\label{sec:simpl-model-cust}
We now derive an approximate model of customer behavior that can be used in the
context of MPC. Customer behavior means their response to price incentives, as
described in (\ref{eq:prob-incentives}) and therefore enters into the system
dynamics $f_t(x(t),u(t))$. However this response is nonlinear, which as
described in the preceding section leads to a non-convex MPC problem
(\ref{eq:mpc1}).

We first explain the origin of this nonlinearity. Consider two neighbor stations
$n', n'' \in N_s$ with equal distance to $s$, for which the incentives offered
from station $s$ are equal, $p_{s,n'} = p_{s,n''}$. If there are customers
equally willing to go to $n'$ or $n''$, an infinitesimally small increase in
$p_{s,n'}$ would cause all those customers to choose $n'$ if we assume they act
totally rationally. The customer reaction to incentives is thus discontinuous in
the prices, and the true behavior model (\ref{eq:prob-incentives}) will not lead to a tractable
optimization problem.

To formulate a tractable MPC problem we approximate $\pi(s,n,p_s)$ in a linear
fashion and choose a convenient set $N_s$ of $N$ nearby neighbors for each
station, so that $|N_s| = N$. The linearization $\bar \pi(s,n,p_s)$ is computed
using Algorithm \ref{alg:gen-lin-model}, which creates samples of customer
reactions to random incentive offers (to all neighboring stations) and a
performs a least-squares fit between observed behavior and the linear model. The
linear dependency on offered incentives, where customers reject station $s$ and
go to neighbor $n$ instead, is defined by vectors $\tilde{\pi}_{s,n}$ of size
$N$, for each $s\in S, n\in N_s$.
\begin{equation}
 \bar \pi(s,n,p_s) = \tilde \pi_{s,n}^\top p_s
\end{equation}
\begin{algorithm}
  \begin{algorithmic}[1]
    \Require $s \in S$,\\
    $N_s$, \Comment{$|N_s| = N$ nearest neighbours around station $s$}\\
    $\Call{taken\_incentive}{s,N_s,p}$, \Comment{Neighbor chosen by the customer as described in Section \ref{sec:cust-react-incent}}\\
    $p_{\max}$, \Comment{Maximum payout}\\
    $P \gg 0$, \Comment{Number of generated payout vectors (samples)}\\
    $C \gg 0$, \Comment{Number of customer behavior samples}\\
    $\Omega$ \Comment{Set of samples. Each sample is a tuple of two vectors: The
      offered payouts $p$ to the $N$ neighbours and the percentage of customers
      taking a certain incentive $\delta$.} \For{$i=1$ to $P$} \State $p
    \leftarrow p_{\max} \cdot \Call{rand}{N} \in [0, p_{\max}]^N$ \Comment{Payout vector}
    \State $e \leftarrow \{0\}^N$ \Comment{Initialize behavior count} \For{$c=1$ to $C$}
    \State $n' \leftarrow \Call{taken\_incentive}{s,N_s,p} \in\{\mathbb N^+,\emptyset\}$
    \State $e_{n'} \leftarrow e_{n'} + 1$
    \EndFor
    \State $\delta \leftarrow e/C$ \Comment{Fraction taking a certain incentive}
    \State $\Omega(i) \leftarrow (p, \delta)$
    \EndFor
    \For{$n=1$ to $N$}
    \State $\tilde{\pi}_{s,n} \leftarrow \underset{\pi}{\arg\,\min}
      \sum_{i\in\{1,\ldots, P\}}\left(\pi^\top \Omega(i)_{p} -
        \Omega(i)_{\delta,n}\right)^2$
    \EndFor
  \end{algorithmic}
  \caption{Fitting the linear customer behavior model}
  \label{alg:gen-lin-model}
\end{algorithm}
\subsection{Computing dynamic price incentives}
\label{sec:mpc}
In this subsection we formulate an MPC problem, the solution of which gives the
price incentives $p_s(t)$ that should be offered to customers for
$t=0,\ldots,T-1$, where $p_s(0)$ are the prices to be issued immediately, and
$p_s(1),\ldots,p_s(T-1)$ are prices planned for subsequent steps.

The number $f_{s}(t)$ of bikes present at station $s$ at time $t$ evolves
according to the original arrival rate $\lambda_s(t)$ and net change $\eta_s(t)$
described in Section \ref{sec:model}, along with a modification
$\gamma(s,\lambda_s(t), p_s(t))$ due to customers taking price incentives and
another, $\Delta f_s(t)$ due to trucks adding or taking away bikes from the
stations. $\tilde N_s$ denotes the set of stations having $s$ as one of their
nearest neighbors.
\begin{equation}
\begin{aligned}
  \gamma(s,\lambda_s(t),p_s(t)) &= \sum_{\tilde n\in \tilde N_s} \pi\left(\tilde n,s,p_{\tilde n}(t)\right) \cdot \lambda_{\tilde n}(t)\\
&- \sum_{n\in N_s} \pi\left(s,n,p_s(t)\right) \cdot \lambda_s(t)
  \end{aligned}
\end{equation}
\begin{equation}
  f_s(t+1) = f_s(t) + \eta_s(t) + \gamma\left(s,\lambda(t), p(t)\right) + \Delta f_s(t).
  \label{eq:lin-model}
\end{equation}
Note that $\sum_{s\in S} \gamma(s,\lambda_s(t), p_s(t))= 0$ since the total number of bikes in
the system must be constant. Also, the controller assumes that customers who
take an incentive go from their originally intended destination to the new one
within the same time step. We do this to make the MPC problem easier to solve,
and assume it does not distort predictions of customer actions too much.

We now specify the components of MPC problem (\ref{eq:mpc1}). Using the
linearized model of customer reactions to incentives from Section
\ref{sec:simpl-model-cust}, and defining quadratic stage costs $c^x_t$ and
$c^u_t$, the MPC problem becomes a quadratic program. Under the assumption of a
linear customer response to prices, the expected payouts are a quadratic
function of the prices, and the input cost in the MPC problem represents a real
monetary cost to the system operator. The state cost aims to penalize loss of
customer service. The resulting MPC problem is a quadratic program (QP) and can
be stated as follows:
\begin{equation}
  \label{eq:mpc}
  \min_{p(t)} \sum_{t = 1}^{\Tprice} \sum_s^S Q_s(t)\tilde f_s(t)^2 + \sum_{t = 0}^{\Tprice-1} \sum_s^S R_s(t) p_s(t)^2
\end{equation}
such that
\begin{IEEEeqnarray}{rCl'l}
  \tilde f_s(t) &=& f_s(t) - \tfrac{1}{2}\left(\underline f^s_t + \overline f^s_t\right),&\forall s\in S, \forall t\IEEEyessubnumber\label{eq:mpc-plateau}\\
f_s(t+1) &=&f_s(t) + \eta_s(t)+ \Delta f_s(t)\nonumber\\
&+& \sum_{\tilde n \in \tilde N_s}\left(\tilde \pi_{\tilde n,s}^\top p_{\tilde n}(t)\right)\lambda_{\tilde n}(t)\nonumber\\
&-&\sum_{n\in N_s}\left(\tilde \pi_{s,n}^\top p_s(t)\right)\lambda_s(t),&\forall s\in S,\, \forall t\IEEEyessubnumber\label{eq:mpc-stateupdate}\\
&& \sum_{n\in N_s} \tilde \pi_{s,n}^\top p_s(t) \leq 1, &\forall s\in S, \, \forall t\IEEEyessubnumber\label{eq:mpc-100p}\\
&&  0\leq p_{s,n}(t) \leq p_\text{max}, &
  \begin{aligned}
    &\forall s\in S,\, \forall t,\\
&n\in \{1,\ldots, N\}
  \end{aligned}\IEEEeqnarraynumspace\IEEEyessubnumber
\label{eq:mpc-maxpayout}
\end{IEEEeqnarray}
The cost weights $Q_s(t)$ and $R_s(t)$ in the cost function (\ref{eq:mpc})
are used to penalize deviation $\tilde f(t)$ from the optimal state
$\frac{1}{2}(\underline f^s_t + \overline f^s_t)$, and the cost caused by the
incentives payout, respectively. A weighting factor $\alpha$ is used to adjust
the relative costs associated with having stations deviate from their optimal
point of operation in the middle of the utility plateau (which we assume to lead
to a lower service level), and cash payouts:
\begin{align}
  Q_{s}(t) &= 1/(\overline f^s_t - \underline f^s_t),\\
  R_{s}(t) &= \alpha \sum_{n\in N_s} \tilde \pi_{s,n}^\top \lambda_s(t).
\end{align}
A lower value of $\alpha$ leads to a lower relative penalty for paid incentives,
likely leading to higher price incentives applied.

Equation (\ref{eq:mpc-plateau}) transforms the number of bikes at each station
to a quantity measured relative to the ``best'' fill level, the middle of the
station's utility plateau. The predicted system states within the horizon are
defined by (\ref{eq:mpc-stateupdate}). It includes the expected arrival and
departure rates as well as the linearized model of customer behavior. Equation
(\ref{eq:mpc-100p}) limits the payouts such that no more than 100\% of arriving
customers take an incentive to one of the neighbours, and
(\ref{eq:mpc-maxpayout}) ensures that payouts are at most $p_\text{max}$.

\section{Simulation}
\label{sec:simulation}

\subsection{Simulation setting}
Based on the assumptions and the system model developed in
Section~\ref{sec:model}, a Monte-Carlo simulation is used to compare the two
approaches for bike repositioning. It is important to note that although
simplified models of the system are used to choose truck actions and prices, we
use the \emph{full} model derived from historical data as described in Section
\ref{sec:model} to simulate the actual behavior of customers. We simulate first
a sequence of weekdays, then a sequence of weekend days, bearing in mind that
demand patterns differ significantly between the two. Every simulation run
consists of a 24h burn-in period starting from the initial system configuration,
in order to reduce the dependence of our results on this initial configuration.
Then, three consecutive days are simulated, for which the statistics gathered
are presented below. In accordance to \cite{redistribution_2012}, redistribution
with trucks is performed during 8am--10pm.

\subsection{Simulation results}

The resulting service level is computed as follows:
\begin{equation}
\text{Service level} = \frac{\text{Potential customers} - \text{No-service events}}{\text{Potential customers}}
\end{equation}
When simulating three consecutive weekdays, about 49,800 potential customers are
generated on average, and for three consecutive weekend days (e.g.~a Bank
Holiday weekend) about 29,900. The number of total no-service events is the sum
of customers who could not rent a bike at an empty station and customers who
wanted to return their bike at a full station.

We varied number of trucks used for repositioning and the level of price
incentives given out (via the choice of state cost weight $\alpha$ in the price
controller). Figure~\ref{fig:servicelevel} shows how the service level reported
by the simulations varied as a result.

\begin{subfigures}
  \begin{figure}
    \begin{minipage}[h]{.475\linewidth}
    \centering
    \includegraphics[width=\textwidth]{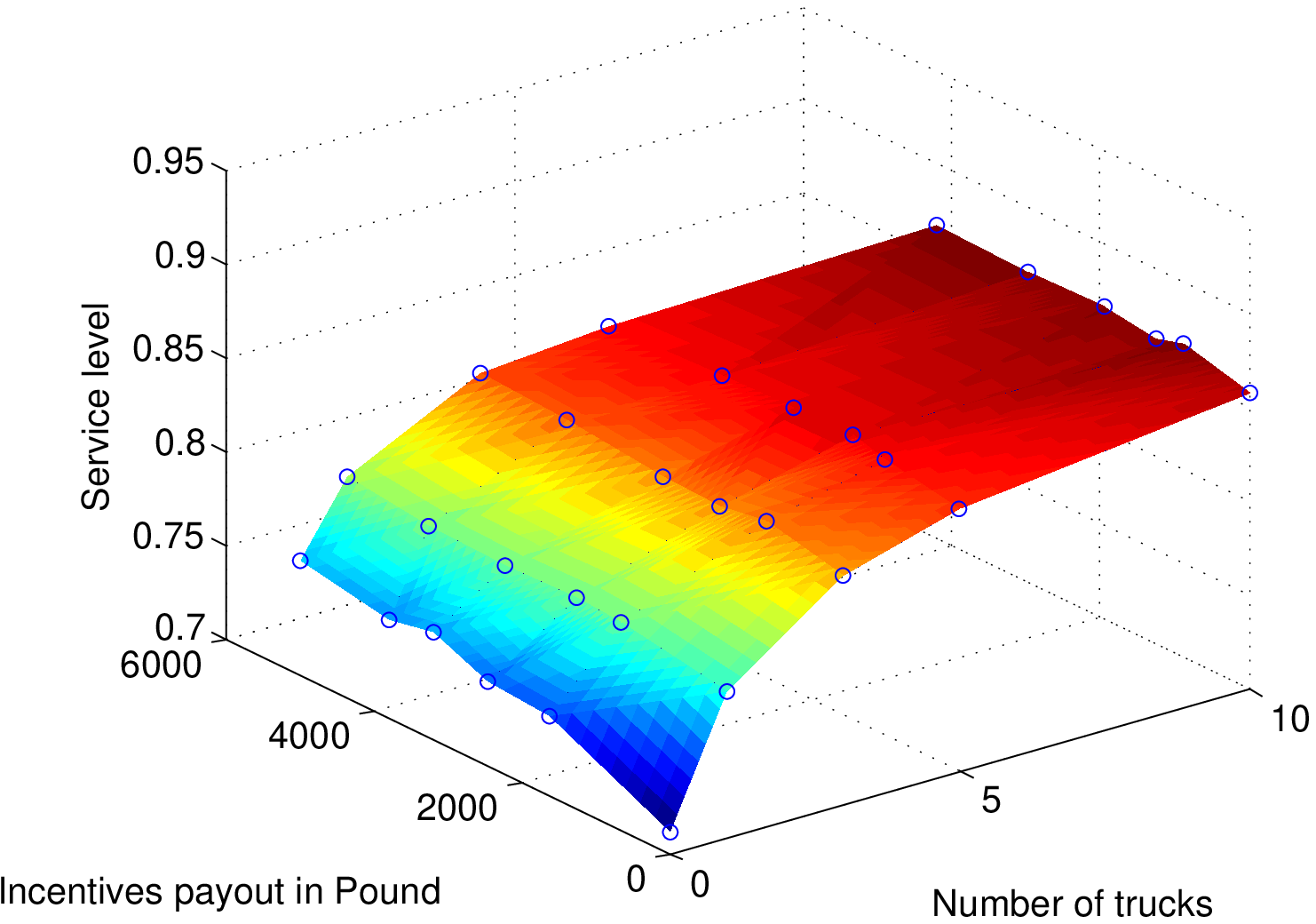}
    \caption{Service level for weekdays, as a function of number of trucks and
      total payouts.}
    \end{minipage}
    \hfill
    \begin{minipage}[h]{.475\linewidth}
    \centering
    \includegraphics[width=\textwidth]{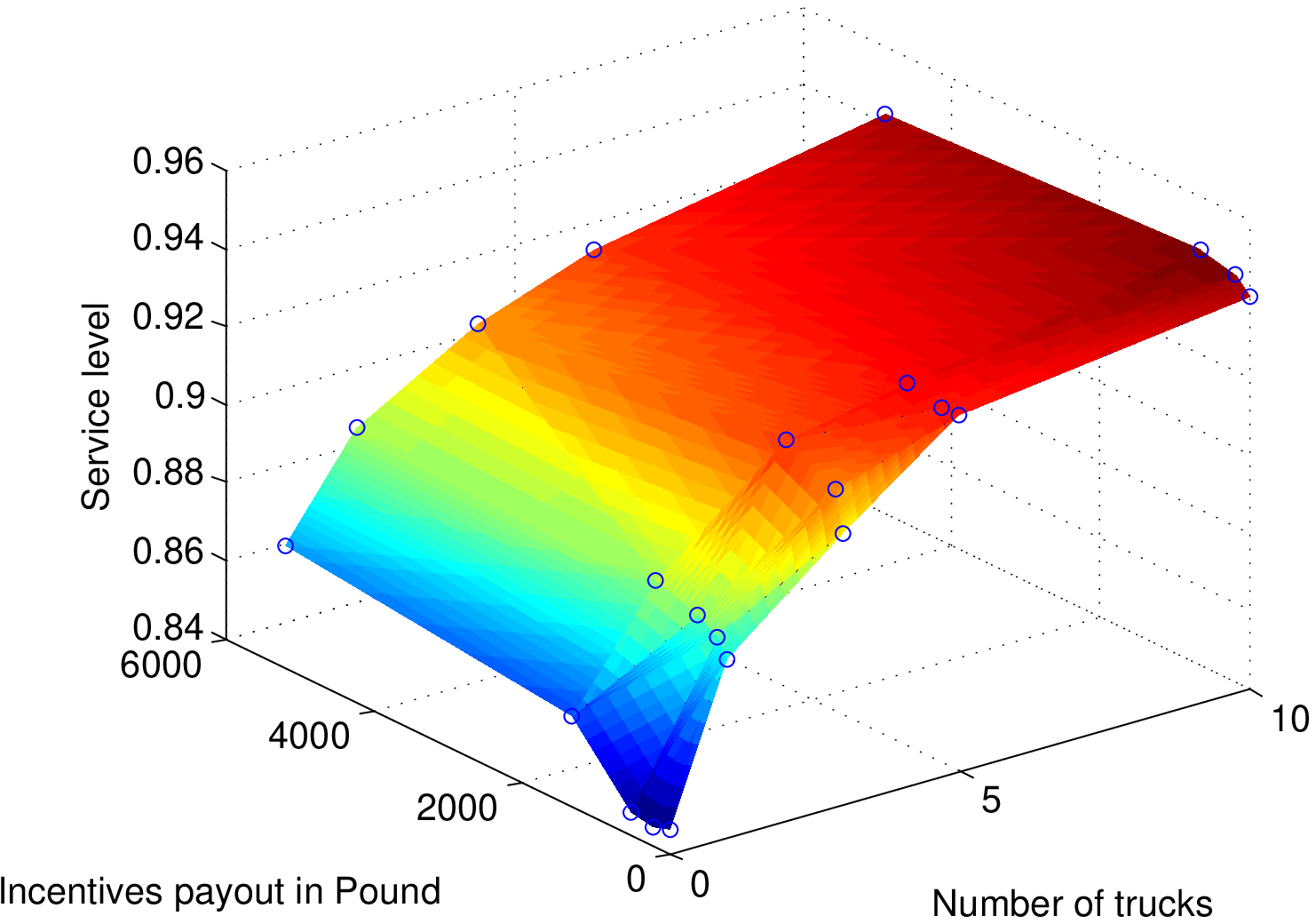}
    \caption{Service level for weekend days, as a function of number of trucks and
      total payouts.}
    \end{minipage}
  \label{fig:servicelevel}
  \end{figure}
\end{subfigures}

\begin{subfigures}
  \begin{figure}
    \begin{minipage}[h]{.475\linewidth}
    \centering
    \includegraphics[width=\textwidth]{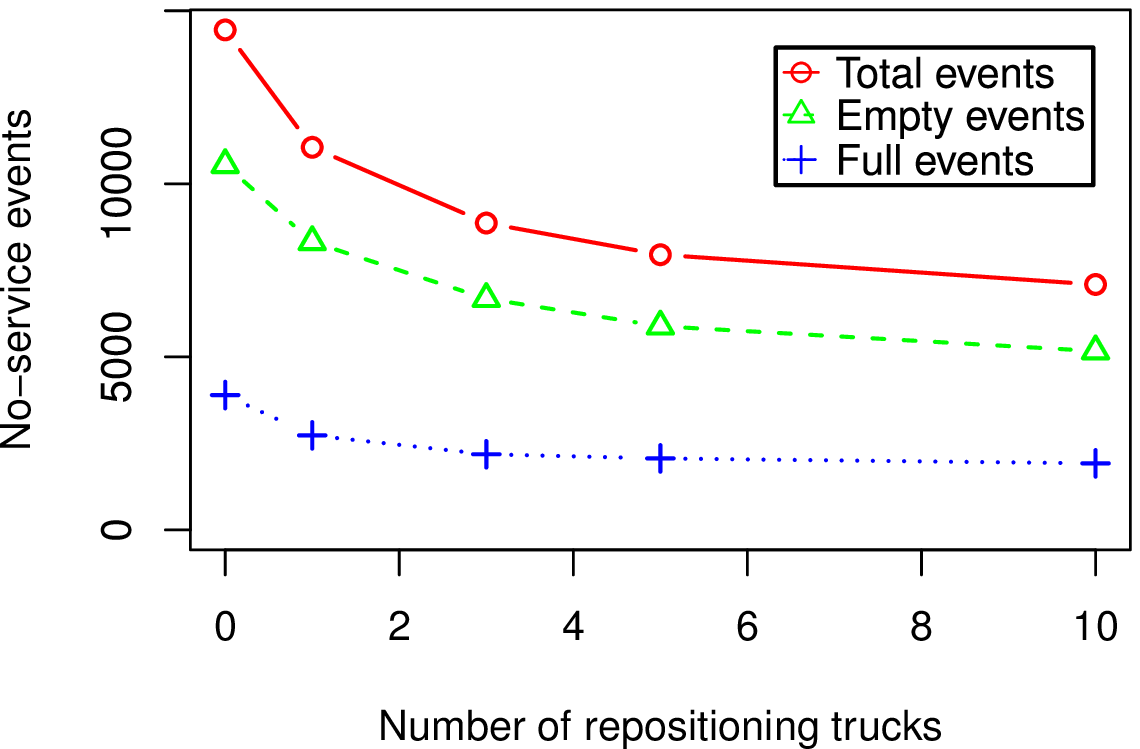}
    \caption{No-service events for different numbers of repositioning trucks (no
      incentives) for three consecutive weekdays}
    \end{minipage}
    \hfill
    \begin{minipage}[h]{.475\linewidth}
    \centering
    \includegraphics[width=\textwidth]{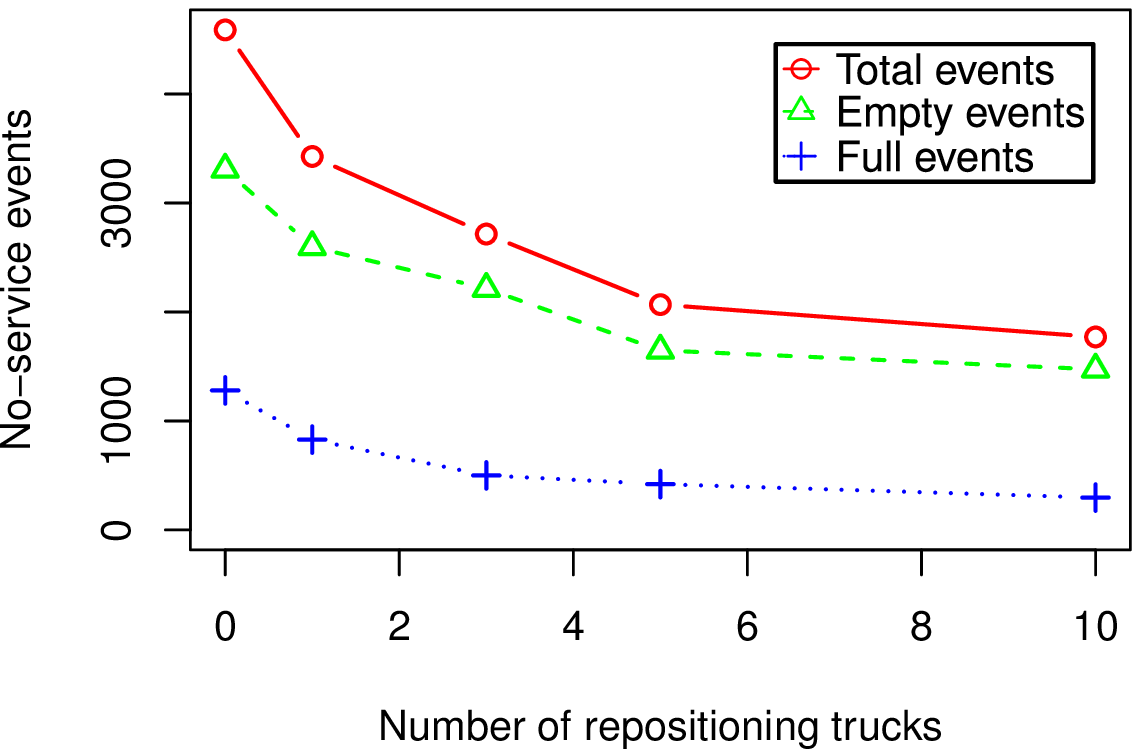}
    \caption{No-service events for different numbers of repositioning trucks (no
      incentives) for three consecutive weekend days}
    \end{minipage}
  \label{fig:trucks}
  \end{figure}
\end{subfigures}

\begin{subfigures}
  \begin{figure}
    \begin{minipage}[h]{.475\linewidth}
      \centering
      \includegraphics[width=\textwidth]{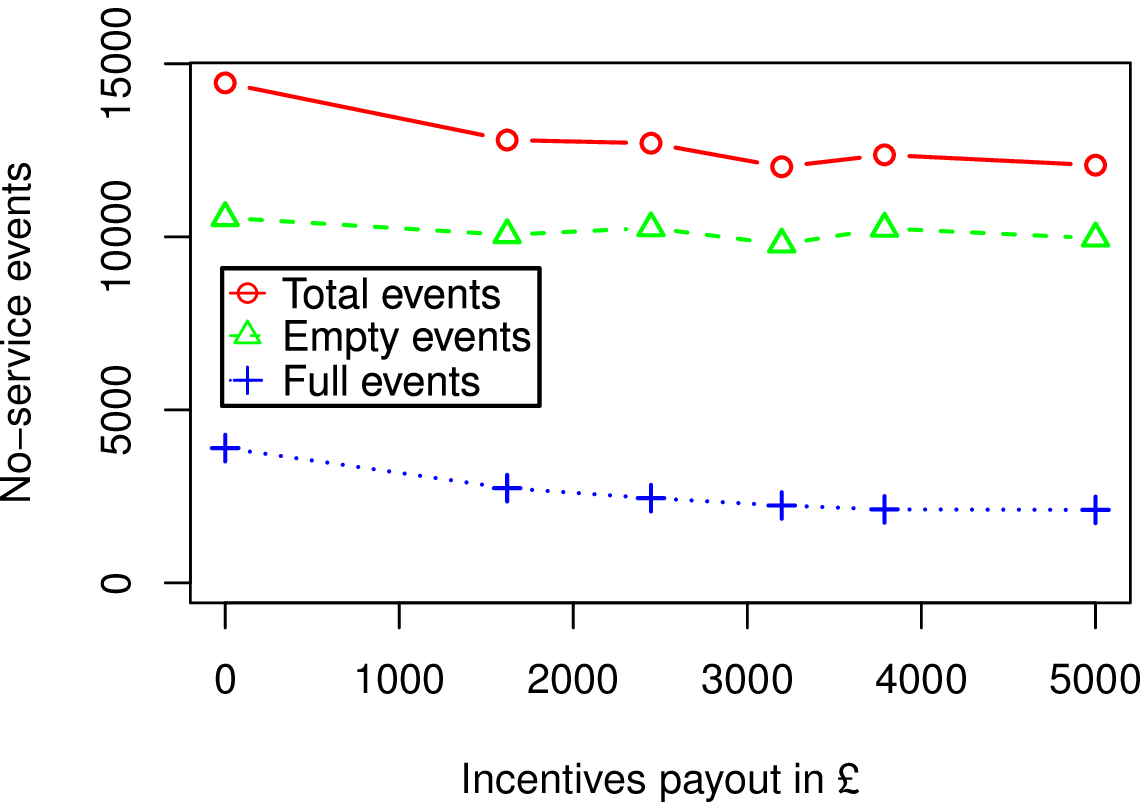}
      \caption{No-service events during simulations for weekdays with 0
        repositioning trucks. Over the course of 72h ca. 49,800 potential
        customers arrive.}
    \end{minipage}
    \hfill
    \begin{minipage}[h]{.475\linewidth}
      \centering
      \includegraphics[width=\textwidth]{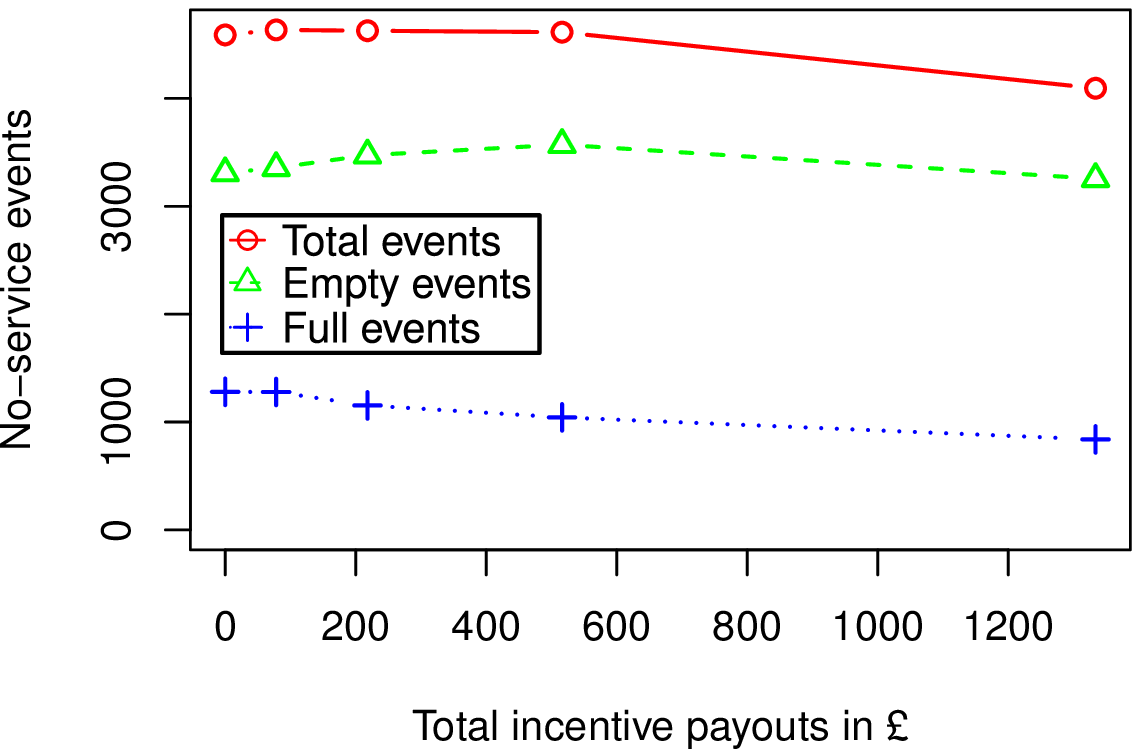}
      \caption{No-service events during simulations of weekend days with 0
        repositioning trucks. Over the course of 72h ca. 29,900 potential
        customers arrive.}
    \end{minipage}
    \label{fig:incentives}
  \end{figure}
\end{subfigures}

As expected, adding more trucks as well as paying out more in incentives has a
positive effect on the service level. However, with increasing service level,
adding trucks and incentives payouts becomes less efficient. Comparing the two
simulations, it appears that the usage peaks caused by commuters were
responsible for most of the service shortfalls observed. Most events where a
customer could not be served were concentrated on only a few
stations.

Figures~\ref{fig:trucks} and~\ref{fig:incentives} show the split of no-service
events into ``empty events'' (where customers wanting to rent a bike arrive at
an empty station) and ``full events'' (where customers wanting to return a bike
arrive at a full station). Since the number of full events is considerably lower
than the number of empty events, it seems plausible that adding more bikes could
have a positive effect on the service rate.

\section{Conclusions}
\label{sec:conclusion}

This paper considered how a Public Bicycle Sharing scheme could be managed using
a combination of intelligently routed repositioning trucks and redistribution
incentives for customers. The truck routes and price incentives were computed
using model-based receding horizon optimization principles, which took account
of expected future customer behavior. As the number of trucks was increased,
diminishing gains to service level were reported for added trucks and customer
incentive payouts. Customer payments were shown to be a means of reducing
service shortfalls, particularly when few repositioning trucks were in
operation.

Our results suggest that price incentives are viable for repositioning bicycles
in a PBS when the commuting rush hour is less prominent. For the London PBS,
price incentives alone were shown to be enough to keep the service level above
87\% on weekends without the use of staff. On weekdays however, when many
customers use the PBS to commute to work, price incentives alone are not
sufficient to lift the service level substantially.

The price control algorithm could be developed further in several ways. Firstly,
a field trial could be used to improve the accuracy of the customer decision
model upon which our controller is based. This would reveal the range of price
elasticities customers exhibit, and also indicate to what extent customer
responses to prices are irrational. Secondly, some simplifying assumptions (for
example deterministic customer arrival, linearized customer reaction to
incentives) could be replaced by more detailed models. However, it is unsure how
much performance can be gained here, as the prediction horizon for optimization
might have to be shortened considerably to account for the increase in
computational complexity.

\bibliographystyle{IEEEtran}      
\bibliography{literature}   

\end{document}